\documentclass[twocolumn,10pt]{IEEEtran}

\usepackage{braket}

\renewcommand{\baselinestretch}{0.83}

\input{def.tex}
\usepackage{dsfont}
\usepackage{minibox}
\DeclareSymbolFont{matha}{OML}{Txmi}{m}{it}
\DeclareMathSymbol{\varv}{\mathord}{matha}{118}
\usepackage{eurosym}
\usepackage{hhline,physics}
\usepackage{algorithm,algorithmicx}
\usepackage{multicol}\IEEEoverridecommandlockouts
\begin{document}
	\title{Ergodic Capacity of IRS-Assisted  MIMO Systems with Correlation and Practical Phase-Shift Modeling}
	\author{Anastasios Papazafeiropoulos \thanks{A. Papazafeiropoulos is with the Communications and Intelligent Systems Research Group, University of Hertfordshire, Hatfield AL10 9AB, U. K., and with the SnT at the University of Luxembourg, Luxembourg. This work was supported by the University of Hertfordshire’s 5-year Vice Chancellor’s Research Fellowship and by the National Research Fund, 	Luxembourg, under the project RISOTTI. E-mail: tapapazaf@gmail.com}}
	\maketitle
	\makeatletter

	\pagestyle{empty}

	\vspace{-2.9cm}
	
	\begin{abstract}
		We focus on the maximization of the exact ergodic capacity (EC)		of a point-to-point multiple-input multiple-output (MIMO) system assisted by an intelligent reflecting surface (IRS). In addition, we  account for the effects of correlated Rayleigh fading and the intertwinement between the amplitude and the phase shift of the reflecting coefficient of each IRS element, which are usually both  neglected despite their presence in practice. Random matrix theory tools allow to derive the  probability density function (PDF) of the cascaded channel in closed form, and subsequently, the EC, which depend only on the large-scale statistics  and the phase shifts. Notably, we optimize  the EC with respect to the phase shifts with low overhead, i.e., once per several coherence intervals instead of the burden of frequent necessary optimization required by expressions being dependent on instantaneous channel information. Monte-Carlo (MC) simulations verify the analytical results and demonstrate the insightful interplay among the key parameters and their impact on the EC.
	\end{abstract}
	
	\section{Introduction}
	Intelligent reflecting surface (IRS) has emerged as a promising green and cost-effective  solution towards the sustainable growth of next-generation wireless networks \cite{Basar2019}. An IRS is a planar meta-surface, which consists of a large number of nearly passive elements that are managed by a smart controller. The role of each element is the induction of an independent phase shift and/or amplitude attenuation to the impinging electromagnetic waves. In \cite{Abeywickrama2020}, it was shown that the   amplitude and phase response are intertwined, while this was neglected by most works. The proper design of the reflecting coefficients enhances the performance by constructive and destructive addition of the signals reflected by IRS to increase the desired signal power and mitigate the co-channel interference, respectively.
	
	The study of IRS performance gains has attracted significant attention under various application scenarios and objectives \cite{Wu2019a, Elbir2020,Papazafeiropoulos2021a,Papazafeiropoulos2021,Wang2021}. For instance, in \cite{Wu2019a},  a minimization of the total transmit power took place by optimizing the transmit and reflecting beamforming (RB). Also, in \cite{Papazafeiropoulos2021a}, we suggested the use of multiple distributed IRSs and compared two practical scenarios, namely, a large number of finite size IRSs and a finite number of large IRSs to show which implementation scenario is more advantageous. Moreover, in \cite{Papazafeiropoulos2021}, we studied the impact of hardware impairments at both the transceiver and IRS sides in a general IRS-assisted multi-user multiple-input single-output (MISO) system with correlated Rayleigh fading, which is usually neglected in relevant works (e.g., see \cite{Basar2019,Wu2019a}), although it should be taken into account \cite{Bjoernson2020}.  Furthermore, in \cite{Wang2021}, the assumptions required to integrate massive multiple-input multiple-output (MIMO) systems with IRS were studied and the achievable rate was obtained. 
	
	In this direction, most existing works on IRS-assisted systems have focused on   single-input single-output (SISO) or MISO systems, while limited research has been devoted to MIMO communication, e.g., see \cite{Zhang2020a,Mu2021,Zhang2021}. In particular, regarding the study of capacity,  the capacity limit of point-to-point  MIMO IRS-assisted systems was studied in \cite{Zhang2020a} and the capacity region  with non-orthogonal multiple access (NOMA) was characterized in \cite{Mu2021}. In \cite{Zhang2021}, the ergodic capacity (EC) was studied for SISO channels.  Moreover, in\cite{Guo2019}, the EC for millimeter wave (mmWave) MIMO systems was investigated but not obtained in closed-form. Similarly, in 	\cite{Yang2020a,Wang2021a},  only upper bounds for the EC were derived. 

	The previous observations motivate the topic of this paper, which is the derivation of the exact  EC of MIMO IRS-assisted systems  for Rayleigh channels under practical considerations while performing robust optimization. Contrary to \cite{Zhang2020a,Mu2021,Zhang2021}, we have  accounted for both correlated fading and the intertwinement between  the amplitude and phase-shift of each element, which have been disregarded in most previous works. Also, \cite{Zhang2021} considered only SISO channels, while we focus on MIMO channels. Furthermore, \cite{Guo2019} focused on mmWave systems and \cite{Yang2020a,Wang2021a} derived just  upper bounds, while these  works did not focus on the exact analysis of the EC. To this end, we have achieved to obtain  the probability density function (PDF) of an IRS-assisted MIMO channel with correlation in a simple closed-form expression,  and we have optimized the EC by considering a practical phase-shift model. Moreover, the optimization in \cite{Zhang2020a,Mu2021} relied on instantaneous channel state information (CSI), while our analytical expression depends only on large-scale statistics as in \cite{Papazafeiropoulos2021a,Papazafeiropoulos2021,Wang2021}, and thus, it can be optimized once per several coherence intervals. Hence, our approach is of great  value since it achieves to reduce considerably the signal overhead, which becomes prohibitive as the number of IRS elements increases. Finally, we shed light on the impact of the fundamental system parameters on the EC. For example, we show its degradation due to correlated Rayleigh fading.
	
	%


	\section{System Model}	
	We consider a MIMO communication system assisted by a  two-dimensional rectangular IRS dynamically adjusted by the IRS controller, where the transmitter  and the receiver  have $ M $ and $ K $, antennas, respectively, as shown in Fig. \ref{Fig0}. The  IRS,  equipped with $ N =N_{\mathrm{H}}N_{\mathrm{V}}$ nearly passive reflecting elements, enables the communication  between the transmitter and the receiver, where  $ N_{\mathrm{H}} $ and  $ N_{\mathrm{V}} $ denote the elements per row   and per column of the IRS. Each element has size $ d_{\mathrm{H}} d_{\mathrm{V}}$, where $ d_{\mathrm{H}} $ is
	its horizontal width and $ d_{\mathrm{V}} $ is the vertical height.
	
	\begin{figure}[!h]
		\begin{center}
			\includegraphics[width=0.85\linewidth]{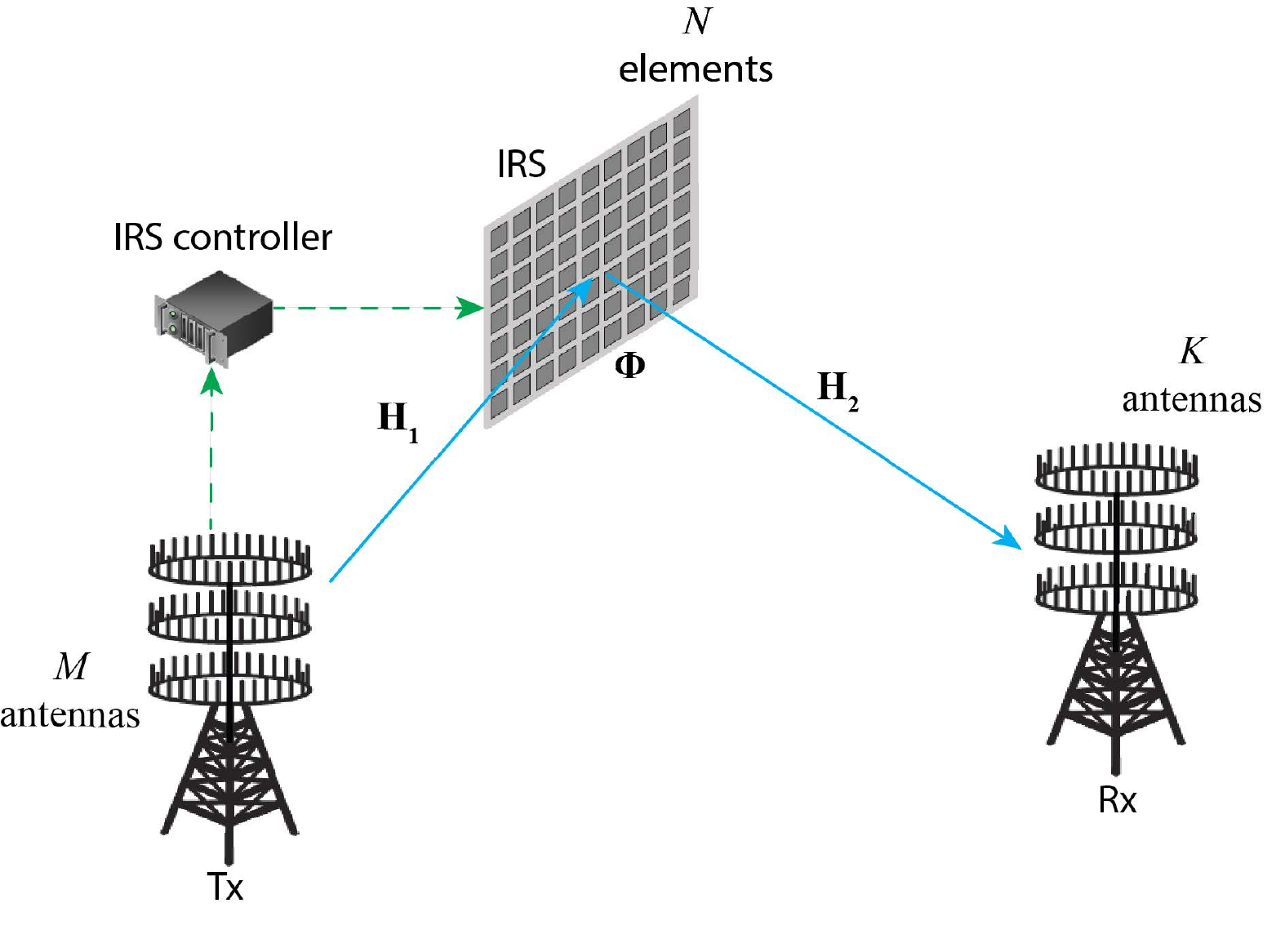}
			\caption{\footnotesize{An IRS-assisted MIMO communication system with $ M $  antennas at the transmitter, $ N $ IRS elements, and $ K $  antennas at the receiver. }}
			\label{Fig0}
		\end{center}
	\end{figure}

	We assume narrowband transmission on  quasi-static block-fading channels, where in each block all the channels remain constant. In particular, let 
	$ \bH_{1}=[\bh_{1,1}\ldots,\bh_{1,M} ] \in \mathbb{C}^{N \times M}$    be the  channel matrix between  the transmitter and the IRS with $ \bh_{1,i} \in \mathbb{C}^{N \times 1}$ for $ i=1,\ldots,M $ being its column vectors. Similarly, $ \bH_{2} \in \mathbb{C}^{K \times N}$ expresses the channel matrix between the IRS and the receiver. The subscripts $ 1$ and $ 2$ correspond to the transmitter-IRS and IRS-receiver  links, respectively. Despite that many  works, e.g., \cite{Wu2019a}, assumed independent Rayleigh fading model, we take into account the spatial correlations of all  nodes, which are unavoidable in practice and affect the performance \cite{Alfano2004,Bjoernson2020}. Especially,  in \cite{Bjoernson2020}, a correlation model, being suitable for IRS, was presented. 
	
	By relying on the Kronecker model to account for the  spatial correlation of the MIMO channel for each link, the  correlated Rayleigh fading distribution for the $ i $th link can be written as
	\begin{align}
		\bH_{i}&=\bR_{i}^{\frac{1}{2}}\bX_{i}\bT_{i}^{\frac{1}{2}},~\mathrm{for}~i=1,2\label{channel}
	\end{align}
	where $ \bR_{1} \in \mathbb{C}^{N \times N} $, $ \bT_{1} \in \mathbb{C}^{M \times M} $ and $ \bR_{2} \in \mathbb{C}^{K \times K} $, $ \bT_{2} \in \mathbb{C}^{N \times N} $ describe the deterministic nonnegative semi-definite correlation
	matrices at the receive, transmit side of link $ 1 $ and $ 2 $, respectively. Basically, $  \bR_{1}$ and $ \bT_{2} $, corresponding to the IRS, are identical. Notably, all the correlation matrices, assumed to be known by the network, can be obtained through  existing estimation methods (see e.g., \cite{Neumann2018}). For the IRS, we adopt the suitable correlation model presented in \cite{Bjoernson2020}, while the  correlations at the transmitter and the receiver correspond to conventional linear antenna array  modeling (see e.g., \cite{Papazafeiropoulos2021}).   Moreover, $ \bX_{1} \in \mathbb{C}^{N \times M}$ and $ \bX_{2} \in \mathbb{C}^{K \times N} $ express the corresponding fast-fading matrices with elements distributed as 	$  \mathcal{CN}\left(0,1\right) $. Also,   $ \beta_{1} $ and $ \beta_{2} $  denote the path-losses  of the transmitter-IRS  and IRS-receiver  links, respectively.

	Let  $\bTheta=\mathrm{diag}\left( \al_{1}e^{j \phi_{1}}, \ldots, \al_{N}e^{j \phi_{N}} \right)\in\mathbb{C}^{N\times N}$ denote the reflecting beamforming ({RBM) being diagonal and representing the response of the $ N $ IRS elements, where $ \phi_{n} \in [-\pi,\pi)$ and $ \al_{n}(\phi_{n})\in [0,1] $ express the phase-shift  and amplitude coefficient of element $ n $, respectively.  Contrary to many previous works that assumed  lossless metasurfaces ($ \al_{n}(\phi_{n})=1 $), we consider that the amplitude coefficient is phase-dependent  based on the load or surface impedances as mentioned in \cite{Abeywickrama2020}. Especially, we consider that the amplitude response  is explicitly expressed as a function of the phase shift as
		\begin{align}
			\al_{n}(\phi_{n})=(1-\kappa_{\min})\left(\frac{\sin(\phi_{n}-\vartheta )+1}{2}\right)^{\!\!\xi}+\kappa_{\min},\label{phaseShiftModel}
		\end{align}
		where $ \kappa_{\min}\ge 0 $ is the minimum amplitude, $ \vartheta \ge 0 $ is the horizontal distance between $ -\pi/2 $ and $ \kappa_{\min} $, which expresses the distance  between $ -\pi/2 $ and the phase shift $ \phi_{n} $ that minimizes $ \sin(\phi_{n}-\vartheta ) $ (i.e., $ \phi_{n}=-\pi/2+\vartheta $), and $ \xi $ is the steepness of the function curve. If $ \kappa_{\min}=1 $ or $ \xi=0 $, \eqref{phaseShiftModel} results in the ideal phase-shift model, i.e., $ 	\al_{n}(\phi_{n})=1 $.
		
		We assume that signals reflected by the IRS two or more times are ignored and that no direct signal exists due to blockage from obstacles such as buildings. Thus, the effective channel matrix is given by $\bH_{2}\bTheta \bH_{1} $. 
		The received signal vector $\by \in \mathbb{C}^{K\times 1} $ by  the receiver is given by
		\begin{align}
			\by=\bH_{2}\bTheta \bH_{1}\bx+\bz,\label{system1}
		\end{align}
		where  $  \bx$ is the transmit signal vector by  the transmitter and $\bz\sim \cC\cN(\b0, \sigma^{2} \bI_{M})$ is the additive white Gaussian noise (AWGN) vector at  the receiver with $ \sigma^{2} $ expressing the average noise power.
		\section{Ergodic Capacity Analysis}
		By assuming that  the transmitter has no channel knowledge while  the receiver  has perfect CSI of both $ \bH_{1} $ and $ \bH_{2} $\footnote{It is worthwhile to mention that by assuming receive RF chains with sensing abilities integrated 	into the IRS, the acquisition of the CSI of individual channels is feasible as suggested recently in \cite{Zheng2021}. On this ground, the  consideration of  the imperfect CSI scenario, which is more practical, is the topic of future work.}, the EC  in (b/s/Hz) of the MIMO-assisted system in \eqref{system1} is expressed as
		\begin{align}
			&	C=\EE\big[\log_{2}\det \big(\Id_{K}+\frac{	\mathrm{SNR}}{M }\bH_{2}\bTheta \bH_{1}\bQ\bH_{1}^{\H}\bTheta ^{\H}\bH_{2}^{\H}\big)\big],\label{capacity1}
		\end{align}
		where the expectation in \eqref{capacity1} is over $ \bH_{1} $, $ \bH_{2} $ while $ \bQ $ denotes the capacity-achieving input covariance, being normalized by its
		energy per dimension and given by
		\begin{align}
			\bQ=\frac{\EE[\bx\bx^{\H}]}{\frac{1}{M}\EE[\|\bx\|^{2}]}.
		\end{align}
		Note that the normalization guarantees that $ \EE[\tr(\bQ)]=M $. Also, we have
		\begin{align}
			\mathrm{SNR}=\frac{\EE[\|\bx\|^{2}]}{\frac{1}{K}\sigma^{2}}.
		\end{align}
		
		Substitution of the channel expressions described by \eqref{channel} into \eqref{capacity1} gives \eqref{capacity21}. Next, by denoting $ \bLambda_{\bR_{i}} $, $ \bLambda_{\bT_{i}} $, and $ \bP $ the diagonal eigenvalue matrices corresponding to the  matrices $\bR_{i}  $, $ \bT_{i} $, and $ \bQ $, the EC can be written as in \eqref{capacity22}. Note that $ \bP $ describes the capacity-achieving power allocation.
		\begin{figure*}
			\begin{align}
				C&=	\EE\big[\!\log_{2}\det \big(\Id_{K}+\frac{	\mathrm{SNR}}{M }\bR_{2}\bX_{2}\bT_{2}^{\frac{1}{2}}\bTheta\bR_{1}\bX_{1}\bT_{1}^{\frac{1}{2}} \bQ\bT_{1}^{\frac{1}{2}}\bX_{1}^{\H}\bR_{1}^{\frac{1}{2}}	\bTheta ^{\H}\bT_{2}^{\frac{1}{2}}\bX_{2}^{\H}\big)\big]\label{capacity21}\\
				&=\EE\big[\log_{2}\det \big(\Id_{K}+\frac{	\mathrm{SNR}}{M }\bLambda_{\bR_{2}}\bX_{2}\bLambda_{\bT_{2}}^{\frac{1}{2}}\bTheta\bLambda_{\bR_{1}}^{\frac{1}{2}}\bX_{1}\bLambda_{\bT_{1}} \bP\bX_{1}^{\H}\bLambda_{\bR_{1}}^{\frac{1}{2}}	\bTheta ^{\H}\bLambda_{\bT_{2}}^{\frac{1}{2}}\bX_{2}^{\H}\big)\big].\label{capacity22}
			\end{align}
			\hrulefill
		\end{figure*}

		It is generally known that    even the derivation of the PDF of the unordered eigenvalue of $ \bH_{i} \bH_{i}^{\H} $, where $ \bH_{i} $ is given as in \eqref{channel}, is a challenging problem that remains unsolved. In the   case of point-to-point MIMO channel, special cases of correlated fading are considered, where correlation is assumed only at either side of  the transmitter or  the receiver  \cite{Alfano2004}. Also, it is known that doubly correlated Rayleigh fading is not amenable to tractable manipulations \cite{Jin2010}. 
		
		Taking these facts into account, we notice that the capacity expression in \eqref{capacity2} includes two random matrices $ \bX_{1} $, $ \bX_{2} $ while double correlated Rayleigh channels are assumed. To tackle this difficulty, we assume correlation only at  the transmitter or the receiver side while we perform singular value decomposition (SVD) to the channel of the other link. Next, we obtain the conditional  PDF of the unordered eigenvalue of $ \bG \bG^{\H} $,  denoted by $ f_{\lambda|\bL} $, and finally, we derive the marginal PDF of the  unordered eigenvalue. The derivation requires to  consider cases, where either $ \bLambda_{\bR_{2}}=\Id_{N} $ or $ \bLambda_{\bT_{1}} \bP=\Id_{M} $ as in \cite{Alfano2004}. Note that,  due to
		the reciprocity of MIMO channel, correlation at either the receiver or transmitter is equivalent. The only difference is that the former case including $ \bLambda_{\bR_{2}} $ is determined completely by the receive correlation while the latter case depends on both the transmit correlation and the power allocation.

		\subsection{Main Results}	
		Let us start with  $ \bLambda_{\bR_{2}} =\Id_{N}$, then after applying the SVD to $ \bX_{1} $, i.e., $ \bX_{1} =\bU_{1}\bD_{1}\bV_{1}^{\H}$, where $ \bD_{1}=\diag\left(\lambda_{1},\ldots, \lambda_{\min(M,N)}\right) $  is a   diagonal matrix  with diagonal elements being the singular values in increasing order while the matrices $ \bU_{1} \in \mathbb{C}^{N \times N}$ and $ \bV_{1} \in \mathbb{C}^{M \times M}$ are unitary and  contain the corresponding eigenvectors, we obtain
		\begin{align}
			C&=\EE\big[\log_{2}\det \big(\Id_{K}+\frac{	\mathrm{SNR}}{M }
			\bX_{2}\bPsi\bX_{2}^{\H}\big) \big],\label{capacity2}
		\end{align}
		where
		\begin{align}
			\!\!\!\!	\!	\bPsi\!=\!\!
			\begin{cases}\diag\!\left(\gamma_{1}\!(\phi_{1})\lambda_{1}^{2}, \ldots, \gamma_{N}(\phi_{N})\lambda_{N}^{2}\right)&\!\!\!\!\!N\le M\\
				\diag\big(\!\gamma_{1}\!(\phi_{1})\lambda_{1}^{2}, \ldots, \gamma_{M}(\phi_{M})\lambda_{M}^{2},\underbrace{\!0,\ldots,0}_{N-M}\!\big)&\!\!\!\!\! N> M
			\end{cases} 
			\!\!\!	\!\label{PC11}
		\end{align}
		with  $ \gamma_{i}(\phi_{i})= p_{i}\al_{i}^{2}(\phi_{i}) \lambda_{\bT_{1},i}\lambda_{\bR_{1},i} \lambda_{\bT_{2},i} $.
		Note that we have taken into account the invariance of $ \bX_{2} $ under left and right unitary transformation and that to achieve capacity (optimal transmit strategy), the
		eigenvectors of $ \bQ $ must coincide with those of the transmit
		correlation $ \bT_{1}$ \cite{Alfano2004}.  Since the
	transmitter has no channel knowledge,  equal
power allocation to each transmit antenna is the most reasonable
strategy, i.e., choosing $ \bQ=\frac{P}{M} \Id_{M} $ (see \cite{Jin2010} and references therein). Hence, the  required EC optimization below concerns only the phase shifts.
More concretely, we can write
		\begin{align}
			C&=\EE\big[\log_{2}\det \big(\Id_{K}+\frac{	\mathrm{SNR}}{M }
			\hat{	\bX}_{2}\hat{\bL}\hat{\bX}_{2}^{\H}\big) \big],\label{capacity3}
		\end{align}
		where $ 	\hat{\bX}_{2} \sim\mathcal{CN}\left(\b0,\Id_{K}\otimes \Id_{q}\right)$ and $ \hat{\bL}= \diag\left(\gamma_{i}(\phi_{i})\lambda_{i}^{2}\right)_{i=1}^{q}$ with $ {q}=\min\left(M,N\right) $.

		Similarly, if  we set $ \bLambda_{\bT_{1}} \bP=\Id_{M} $  and apply the SVD  to $ \bX_{2} $, which means $ \bX_{2} =\bU_{2}\bD_{2}\bV_{2}^{\H}$, where $ \bD_{2}=\diag\left(\tilde{\lambda}_{1},\ldots, \tilde{\lambda}_{\min(K,N)}\right) $  is a   diagonal matrix  with diagonal elements being the singular values in increasing order while    $ \bU_{2} \in \mathbb{C}^{K \times K}$ and  $ \bV_{2} \in \mathbb{C}^{N \times N}$ are unitary matrices, we obtain
		\begin{align}
			C&=\EE\big[\log_{2}\det \big(\Id_{M}+\frac{	\mathrm{SNR}}{M }\bX_{1}^{\H}\tilde{\bPsi}\bX_{1}\big) \big],\label{capacity4}
		\end{align}
		where
		\begin{align}
			\!\!\!\!	\!	\tilde{\bPsi}\!=\!\!
			\begin{cases}\diag\!\left(\tilde{\gamma}_{1}\!(\phi_{1})\tilde{\lambda}_{1}^{2}, \ldots, \tilde{\gamma}_{N}(\phi_{N})\tilde{\lambda}_{N}^{2}\right)&\!\!\!\!\!\!N\le K\\
				\diag\big(\!\tilde{\gamma}_{1}\!(\phi_{1})\tilde{\lambda}_{1}^{2}, \ldots, \tilde{\gamma}_{K}(\phi_{K})\tilde{\lambda}_{K}^{2},\underbrace{\!0,\ldots,0}_{N-K}\!\big)&\!\!\!\!\! N> K
			\end{cases} 	\!\!\!	\!\label{PC12}
		\end{align}
		with  $ \tilde{\gamma}_{i}(\phi_{i})= p_{i}\al_{i}^{2}(\phi_{i}) \lambda_{\bR_{1},i}\lambda_{\bR_{2},i} \lambda_{\bT_{2},i} $ and the invariance of $ \bX_{1} $ considered under left and right unitary transformation. In a similar way to \eqref{capacity3}, the capacity can be described by
		\begin{align}
			C&=\EE\big[\log_{2}\det \big(\Id_{M}+\frac{	\mathrm{SNR}}{M }
			\tilde{\bX}_{1}^{\H}\tilde{\bL}\tilde{\bX}_{1}\big) \big],\label{capacity5}
		\end{align}
		where $ 	\tilde{\bX}_{1} \sim\mathcal{CN}\left(\b0,\Id_{M}\otimes \Id_{\tilde{q}}\right)$ and $\tilde{ \bL}= \diag\big(\tilde{\gamma}_{i}(\phi_{i})\tilde{\lambda}_{i}^{2}\big)_{i=1}^{\tilde{q}}$ with $ \tilde{q}=\min\left(K,N\right) $.
		
		We observe that \eqref{capacity3} and \eqref{capacity5} have a similar expression. Hence, it is sufficient to  compute  one of them, e.g., \eqref{capacity3}. We can rewrite it in an equivalent way as
		\begin{align}
			\hat{C}&={s}\!\int_{0}^{\infty}\!\log_{2}\big(1+\frac{	\mathrm{SNR}}{M }
			\hat{\lambda}\big) f_{\hat{\lambda}}(\hat{\lambda})\mathrm{d}\hat{\lambda},\label{capacity6}
		\end{align}
		where $ {s}=\min(K,q) $, $ \hat{\lambda} $ is an unordered eigenvalue
		of the random matrix $ \hat{\bX}_{2}\hat{\bL}\hat{\bX}_{2}^{\H} $, and $ f_{\hat{\lambda}}(\hat{\lambda}) $ is its PDF. Notably, there is no expression for  the distribution of   $ \hat{\lambda} $. We notice that \eqref{capacity6} resembles \cite[Eq. 11]{Jin2010} but includes a different form of the diagonal matrix, which can not result in our expressions as a special case.  Remarkably, \eqref{capacity6} provides the EC an IRS-assisted MIMO channel with correlation in terms of large-scale statistics. 
		
		For this reason, it is required  first to obtain a new exact closed-form expression for the joint PDF of $ 0\le c_{1}<\cdots<c_{q}< \infty $, where $ c_{i}=\gamma_{i}(\phi_{i})\hat{\lambda}_{i}^{2} $ are the elements of $ \hat{\bL} $.
		
		\begin{lemma}\label{JointPDFLemma}
			The joint pdf of  $ \hat{\bL} =\diag\left(c_{1}, \ldots,c_{q}\right) $, where  $ 0\le c_{1}<\cdots<c_{q}< \infty $, is given by
			\begin{align}
				f_{ \hat{\bL} }\left(c_{1}, \ldots,c_{q}\right)=	\mathcal{K}\prod_{i=1}^{q}\frac{c_{i}^{p-q}e^{-\frac{c_{i}}{\gamma_{i}}}}{\gamma_{i}^{p-q+1}}\prod_{i<j}^{q}\left(\frac{c_{j}}{\gamma_{j}}-\frac{c_{i}}{\gamma_{i}}\right)^{2},\label{JointPdf}
			\end{align}
			where $ \mathcal{K}=\left(\prod_{i=1}^{q}\Gamma\left(q-i+1\right)\Gamma\left(p-i+1\right)\right)^{-1} $ with  $ p=\max(K,N) $.
		\end{lemma}
		\begin{proof}
			See Appendix~\ref{Lemma1}.
		\end{proof}
		
		Taking into account  Lemma \ref{JointPDFLemma}, we obtain $ f_{\hat{\lambda}}(\hat{\lambda}) $ according to the following theorem.

		\begin{theorem}\label{MarginalPDF}
			The marginal PDF of an unordered eigenvalue $ \hat{\lambda} $ of $ \bX_{1}^{\H}\tilde{\bL}\bX_{1} $ 	is given by 
		\end{theorem}
		\begin{align}
			\!f_{\hat{\lambda}}(\hat{\lambda})\!=\!\frac{\mathcal{K}}{s\!\prod_{i<j}^{q} \!\left(\gamma_{j}\!-\!\gamma_{i}\right)}&\sum_{l=1}^{q}\sum_{k=q-s+1}^{q}\frac{\hat{\lambda}^{(M+p-2q+2k+l-3)/2}}{\Gamma\!\left(M\!-\!q\!+\!k\right)\gamma_{l}^{\frac{p-M+l-1}{2}}}\nn\\
			&\times\mathrm{K}_{p-M+l-1}(2\sqrt{\hat{\lambda}/\gamma_{l}}) G_{l,k},\label{TheoremLambda}
		\end{align}
		where $ G_{l,k} $ is the $ (l,k)th $ cofactor of a $ q \times q $ matrix whose  $ (m,n)th $ 	entry is $ \{\bG\}_{m,n}=\Gamma(p-q+m+n-1)$.
		\begin{proof}
			See Appendix~\ref{Theorem1}.
		\end{proof}
		\begin{remark}
		Theorem \ref{MarginalPDF} presents the PDF of an IRS-assisted MIMO channel with correlation in a simple closed-form expression. Based on this theorem,  the EC in \eqref{capacity6} depends on the number of antennas at  the transmitter, the receiver, and the number or IRS elements together with the eigenvalues of the correlation matrices. A notable dependence in \eqref{TheoremLambda} concerns  the phase-shift and amplitude  in terms of $ \al_{i}^{2}(\phi_{i}) $. 
		\end{remark}
		\subsection{RBM optimization}\label{3b}
		The EC  can be maximized  in terms of the phase shifts. Relying on the common assumption of infinite resolution phase shifters, we formulate the optimization problem as
		\begin{align}\begin{split}
				&\!\!\!(\mathcal{P}1)~\max_{\bPhi} ~~	\hat{C}\\
				&~~~~~~\mathrm{s.t}~-\pi \le \phi_{n}\le \pi,~~ n\!=\!1,\dots,q,
			\end{split}\label{Maximization} 
		\end{align}
		where $ \hat{C}$ is given by \eqref{capacity6}  based on  \eqref{capacity3} or \eqref{capacity5}. 
		
		The problem $ 	(\mathcal{P}1) $ concerns a non-convex maximization  with respect to $  \phi_{n} $.  Taking the expression of $ f_{\hat{\lambda}}(\hat{\lambda}) $ into account, the optimization of \eqref{capacity6} is basically a constrained maximization problem whose  solution could be given by means of the  projected gradient 	ascent until convergence to a 	stationary point  \cite{Papazafeiropoulos2021}. The convergence is guaranteed due to the power constraint.  Since the dependence on phase shifts is hidden only in $ f_{\hat{\lambda}}(\hat{\lambda}) $, below, we focus on their optimization. According to the algorithm, let $ \bs^{i} =[\phi_{1}^{i}, \ldots, \phi_{q}^{i}]^{\T}$ denote 	the  vector including the phases at step $ i $.  The next iteration point  results in the increase of $ f_{\hat{\lambda}}(\hat{\lambda}) $ towards to its convergence. Hence, we have
		\begin{align}
			\tilde{\bs}^{i+1}&=\bs^{i}+\mu \bq^{i},\label{sol1}\\
			\bs^{i+1}&=\exp\left(j \arg \left(\tilde{\bs}^{i+1}\right)\right),\label{sol2}
		\end{align}
		where $ \mu $ is the step size and $ \bq^{i} $ is the adopted ascent direction at step $ i $ with  $ [\bq^{i}]_{n}= \pdv{f_{\hat{\lambda}}(\hat{\lambda})}{\phi_{n}} $. This derivative is provided by  Lemma \ref{Prop:optimPhase} below. 	The solution is found  by formulating the projection problem $ \min_{|\phi_{n} |=1, n=1,\ldots,N}\|\bs-\tilde{\bs}\|^{2} $ based on \eqref{sol1} and \eqref{sol2} under the constraint in \eqref{Maximization}. Note that the suitable step size requires  computation at each iteration, which is achieved by means of the backtracking line search \cite{Boyd2004}.

		\begin{lemma}\label{Prop:optimPhase}
			The derivative of $ f_{\hat{\lambda}}(\hat{\lambda})$ with respect to $ \phi_{n} $ is provided by
			\begin{align}
				&\!\pdv{f_{\hat{\lambda}}(\hat{\lambda})}{\phi_{n}}\!=\!-\frac{ p_{n}\al_{n}(\phi_{n}) \lambda_{\bT_{1},n}\lambda_{\bR_{1},n} \lambda_{\bT_{2},n}\mathcal{K}}{s\!\prod_{i<j}^{q} \!\left(\gamma_{j}\!-\!\gamma_{i}\right)}(1-\kappa_{\min})\xi\nn\\
				&\times\cos(\phi_{n}-\vartheta )\left(\frac{\sin(\phi_{n}-\vartheta )+1}{2}\right)^{\!\!\xi-1}
				\big(\tr\big(\bV^{-1}\pdv{\bV}{\gamma_{n}}\big)\nn\\
				&\times\!\!\sum_{l=1}^{q}\sum_{k=q-s+1}^{q}\!\!\frac{\hat{\lambda}^{(M+p-2q+2k+l-3)/2}}{\Gamma\!\left(M\!-\!q\!+\!k\right)\!\gamma_{l}^{\frac{p\!-\!M\!+\!l\!-\!1}{2}}}\mathrm{K}_{p-M+l-1}(2\sqrt{\!\hat{\lambda}/\gamma_{l}}) G_{l,k}\nn\\
				&+	\sum_{k=q-s+1}^{q}\!\!\frac{\hat{\lambda}^{(2M+\bar{p}-2q+2k-2)/2}}{\Gamma\!\left(M\!-\!q\!+\!k\right)\!\gamma_{n}^{\frac{\bar{p}}{2}}}G_{n,k}\big(\frac{\bar{p}}{2\gamma_{n}}			
				\mathrm{K}_{\bar{p}}(2\sqrt{\!\hat{\lambda}/\gamma_{n}}) \nn\\
				&-\frac{1}{\gamma_{n}}\sqrt{\frac{\hat{\lambda}}{\gamma_{n}}}\big(\mathrm{K}_{\bar{p}-1}(2\sqrt{\!\hat{\lambda}/\gamma_{n}}) +\mathrm{K}_{\bar{p}+1}(2\sqrt{\!\hat{\lambda}/\gamma_{n}})\!\big)\!\big)\!\big)
				,\label{TheoremLambdaDerivative}
			\end{align}
			where $ \bar{p}=p-M+n-1 $, $ \bV $ is the Vandermonde matrix satisfying $ \det(\bV)=\prod_{i<j}^{q} \left(\gamma_{j}-\gamma_{i}\right)$, and $ \pdv{\bV}{\gamma_{n}} $ is a matrix whose elements are the derivatives of   $ \bV $ with respect to $ \gamma_{n} $.
		\end{lemma}
		
		\begin{proof}The computation of the derivative is straightforward since  $ f_{\hat{\lambda}}(\hat{\lambda})$ is in closed-form. We notice that the dependence on the phase shifts is inside $ \gamma_{l} $ in terms of $ \al^{2}_{l}(\phi_{l}) $. Hence, the derivative is obtained as in \eqref{TheoremLambdaDerivative} by using the derivative of \eqref{phaseShiftModel}, \cite[Eq. 46]{Petersen2012}, which gives the derivative of a determinant, and   \cite[Eq. 03.04.20.0014.02]{Wolfram1}, which gives the derivative of $ K_v(\cdot) $ after some simple algebraic manipulations.
		\end{proof}

		\section{Numerical Results}\label{Numerical} 
		We consider  uniform linear arrays (ULAs) for  the configuration of both the transmitter and the receiver, while the IRS consists of a uniform planar array (UPA).  Based on  the 3GPP Urban Micro (UMi) scenario from TR36.814 for a carrier frequency of $ 2.5 $ GHz and noise level $ -80 $ dBm,  the path losses for $ \bH_{1} $ and $ \bH_{2} $ are generated based on the NLOS  version \cite{3GPP2017}. Specifically, therein, the overall path loss for the IRS-assisted
		link is $ 	\beta=	\beta_{1}	\beta_{2} $, where 		$ \beta_{i}=C_{i} d_{i}^{-\nu_{i}},~i=1,2 $		with $ C_{1}=26 $ dB, $ C_{2}=28 $ dB, $ \nu_{1} =2.2$, $ \nu_{2} =3.67$. The variables $ d_{1}=8~\mathrm{m} $ and $ d_{2} =60~\mathrm{m}$ express the distances between the transmitter and the IRS, and the IRS and the receiver, respectively. As can be seen, since the IRS is closer to the transmitter, the path-loss exponent of this link is lower because fewer obstacles are expected. We use $ 5 $ dBi antennas at the transmitter ($ M=4 $) and the receiver ($ K=4 $), and their   correlation matrices  are given by \cite{Papazafeiropoulos2021}. Unless otherwise stated, we consider the following values. The correlation matrix for the  IRS is given by \cite{Bjoernson2020}, where 
		$ d_{\mathrm{H}}\!=\!d_{\mathrm{V}}\!=\!\lambda/4 $.  The phase shift model parameters are $ \kappa_{\min}=0.8 $, $ \xi=1.6 $, and $  \vartheta=0.43\pi $ \cite{Abeywickrama2020}.
		The figures correspond to  $ \bLambda_{\bR_{2}} =\Id_{N}$, while similar observations can be extracted in the case  $ \bLambda_{\bT_{1}} \bP=\Id_{M} $ since we do not focus on the transceiver design but on the impact of the IRS.  Furthermore,  Monte-Carlo (MC) simulations  coincide with the analytical results in all cases, which corroborates our analysis.
		
		\begin{figure}[!h]
			\begin{center}
				\includegraphics[width=0.85\linewidth]{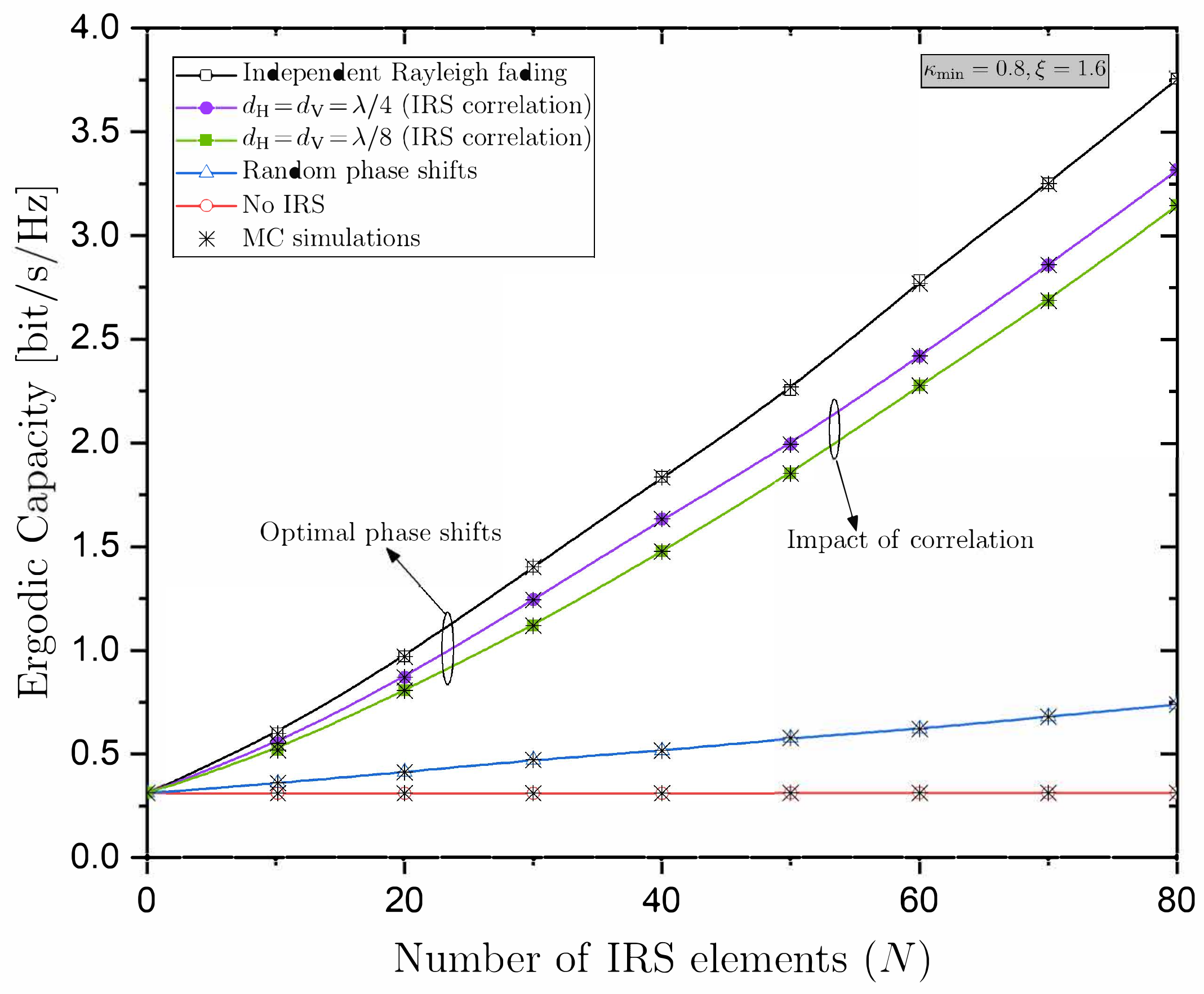}
				\caption{\footnotesize{EC  of an IRS-assisted MIMO system versus the number of IRS elements $ N $ for varying IRS correlation ( $ \kappa_{\min}=0.8 $, $ \xi=1.6 $). }}
				\label{Fig1}
			\end{center}
		\end{figure} 	
		
		\begin{figure}[!h]
			\begin{center}
				\includegraphics[width=0.85\linewidth]{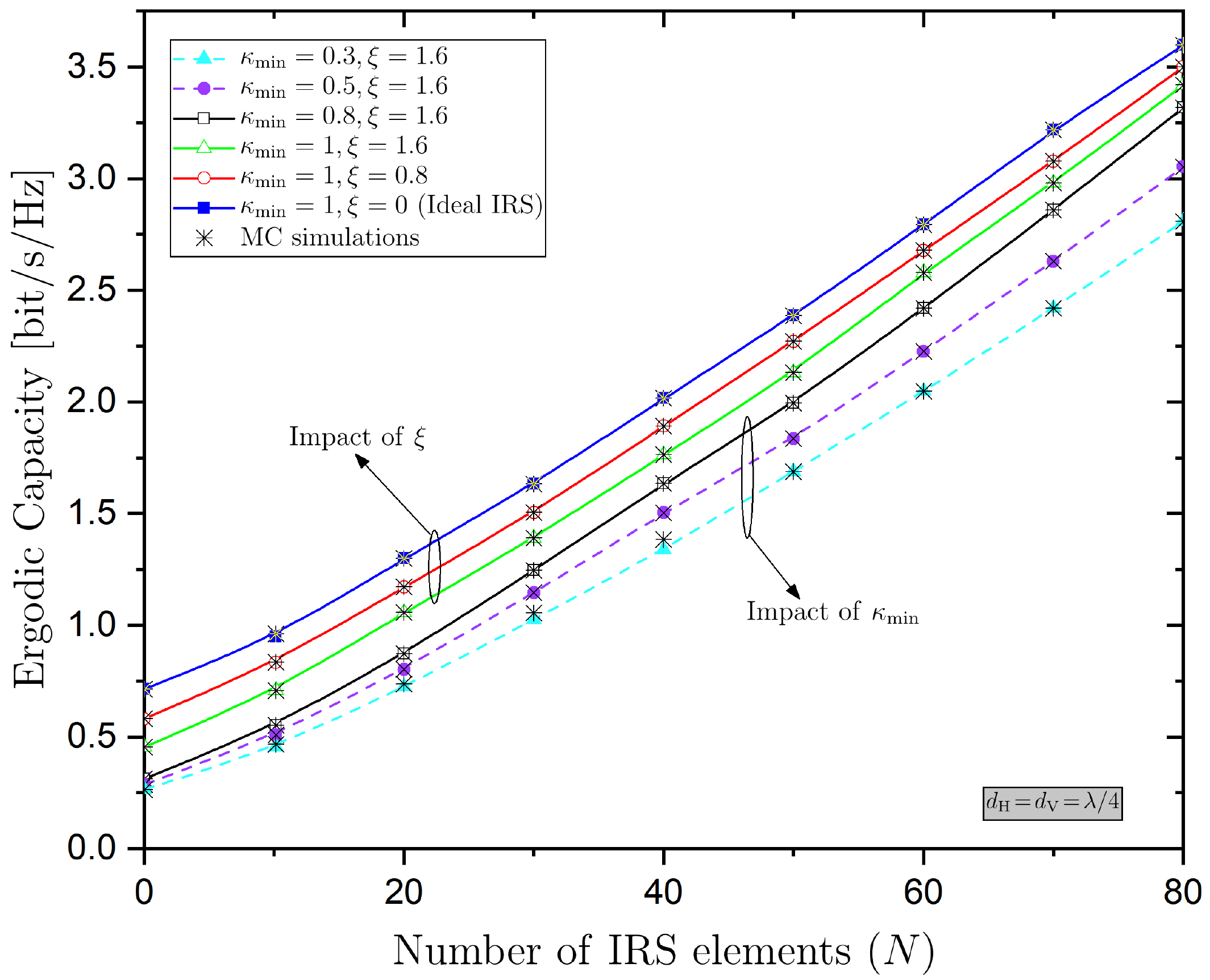}
				\caption{\footnotesize{EC  of an IRS-assisted MIMO system versus the number of IRS elements $ N $ for varying phase-shift model parameters  ( $ \mathrm{d_{\mathrm{H}}=d_{\mathrm{V}}=\lambda/8} $).}}
				\label{Fig2}
			\end{center}
		\end{figure} 	
		Fig. \ref{Fig1} illustrates the EC versus the number of IRS elements. In general, we observe that $ 	\hat{C} $ increases with the number of IRS elements. Moreover, given  specific values for the phase-shift model, we study the impact of the  IRS correlation.  We observe that the unrealistic independent Rayleigh fading gives the highest $ 	\hat{C} $, while the decrease of the distance between the IRS elements (higher correlation) results in a lower capacity. Also, we depict the scenarios with no IRS and with random phase shifts, i.e., not optimized. The former corresponds to a horizontal line as $ N $ increases since the capacity is independent of $N  $, while the latter shows that the phase-shift optimization is clearly very beneficial during the IRS implementation.   
		
		In Fig.  \ref{Fig2}, we depict the EC  for different phase-shift settings. The ``dashed'' and ``solid'' lines correspond to variation of $ \kappa_{\min}$ and $ \xi $, respectively. For the sake of comparison, we have also shown the case of ideal phase-shift model ( $ \kappa_{\min} =1$, $ \xi=0 $). We observe that when  $ \kappa_{\min}$ increases,  EC increases. On the contrary, when  $ \xi $ increases, $ 	\hat{C} $ decreases. Notable, these observations coincide with the results in \cite{Zhang2021} for SISO-IRS assisted channels.

		\section{Conclusion} \label{Conclusion} 
		In this paper, we derived the  PDF  and the EC of a point-to-point  MIMO system assisted by an IRS under the unavoidable realistic conditions of correlated Rayleigh fading and correlation between the amplitude and the phase-shift of each IRS element. Also, we provided the maximization of the EC with  respect to the phase shifts of the IRS elements. This maximization is advantageous since it is accompanied by reduced overhead and computational complexity by exploiting the dependence of the EC by just large-scale statistics.  The analytical results are verified by MC simulations and show how the various system parameters affect the EC. For example, correlated fading results in performance degradation.  Future works on EC  should  account  for Rician fading conditions, and possibly, the existence of a direct path between the transmitter and the receiver.
		
		\begin{appendices}
			\section{Proof of Lemma~\ref{JointPDFLemma}}\label{Lemma1}
			The joint PDF of $ \bW=\diag\left(w_{1}, \ldots,w_{q}\right) $, where $ w_{i}=\hat{\lambda}_{i}^{2} $ and $ p=\max(K,N) $, is given by \cite{James1964} as
			\begin{align}
				\!\!\!\!	f_{\bW}\!\left(w_{1}, \ldots,w_{q}\right)\!=\!\mathcal{K}e^{-\sum_{i=1}^{q}\!w_{i}}\!\prod_{i=1}^{q}\!w_{i}^{p-q}\prod_{i<j}^{q}\!\!\left(w_{j}-w_{i}\right)^{2}\!. \label{JointPDFLemma1}
			\end{align}
			
			Taking into account that $\hat{\lambda}_{i}^{2}= c_{i}/\gamma_{i}(\phi_{i}) $, we obtain the joint PDF of $  \hat{\bL}=\diag\left(c_{1}, \ldots,c_{q}\right) $ after applying a vector transformation to \eqref{JointPDFLemma1} as
			\begin{align}
				f_{  \hat{\bL}}\!\left(c_{1}, \ldots,c_{q}\right)&=	f_{ \bW}\!\left(\frac{c_{1}}{\gamma_1}, \ldots,\frac{c_{q}}{\gamma_q}\right)\nn\\
				&\times|\bJ\left( \left(w_{1}, \ldots,w_{q}\right)\to \left(c_{1}, \ldots,c_{q}\right)\right)|,\label{JointPDFLemma2}
			\end{align}
			where the Jacobian transformation  is evaluated as
			\begin{align}
				|\bJ\left( \left(w_{1}, \ldots,w_{q}\right)\to \left(c_{1}, \ldots,c_{q}\right)\right)|=\prod_{i=1}^{q}\frac{1}{\gamma_{i}}.\label{JointPDFLemma3}
			\end{align}
			
			Substitution of \eqref{JointPDFLemma1} and \eqref{JointPDFLemma3} into \eqref{JointPDFLemma2} gives the joint PDF of $  \hat{\bL} $ in \eqref{JointPdf}.
			\section{Proof of Theorem~\ref{MarginalPDF}}\label{Theorem1}
			
			The unconditional PDF of an unordered eigenvalue of $ \bX_{1}^{\H}\tilde{\bL}\bX_{1} $  is obtained by taking  the expectation over $ \tilde{\bL} $, i.e., $ \EE_{\tilde{\bL}}[f_{\hat{\lambda}|\tilde{\bL}}(\hat{\lambda})] $. For this reason, we employ the  conditional unordered eigenvalue PDF from \cite[Lemma 1]{Jin2010} after expressing it as
			\begin{align}
				\!\!\!\!	f_{\hat{\lambda}|\tilde{\bL}}(\hat{\lambda})\!=\!\frac{1}{s\prod_{i<j}^{q}\!\left(c_{j}\!-\!c_{i}\right)}\!\!\sum_{k=q-s+1}^{q}\!\!\frac{\hat{\lambda}^{M-q+k-1}}{\Gamma\!\left(M\!-\!q\!+\!k\right)}\det\!\left(\bD_{k}\right)\!,
			\end{align}
			where $ \bD_{k} $ is a $ q\times q $ matrix with entries
			\begin{align}
				\{\bD_{k}\}_{m,n}=\begin{cases}c_{m}^{n-1},& ~~~n \ne k, \\
					e^{-\hat{\lambda}/c_{m}}c_{m}^{q-M-1},& ~~~n= k. \end{cases}
			\end{align}
			The unconditional PDF is derived as
			\begin{align}
				f_{\hat{\lambda}}(\hat{\lambda})&=\frac{\mathcal{K}}{s}\sum_{k=q-s+1}^{q}\frac{\hat{\lambda}^{M-q+k-1}}{\Gamma\left(M-q+k\right)}\mathcal{I}_{k},\label{fl}
			\end{align}
			where 
			\begin{align}
				\mathcal{I}_{k}&\!=\!\!\!\displaystyle\int_{0\le c_{1}<\cdots<c_{q}< \infty}\!\!\!\!\!\!\!\!\!\!\!\!\!\!\det\!\left(\bD_{k}\right)\!	\prod_{l=1}^{q}\!\frac{c_{l}^{p-q}e^{-\frac{c_{l}}{\gamma_{l}}}}{\gamma_{l}^{p-q+1}}\!\prod_{i<j}^{q}\frac{\!\left(\frac{c_{j}}{\gamma_{j}}\!-\!\frac{c_{i}}{\gamma_{i}}\!\right)^{2}}{\left(c_{j}-c_{i}\right)}\mathrm{d}c_{1}\cdots\mathrm{d}c_{q}.\nn
			\end{align}
			
			Meanwhile, according to the matrix theory, it holds that $ \prod_{i<j}^{q} \left(\theta_{j}-\theta_{i}\right)=\det(\bX)$, where  $ \bX $ is a Vandermonde matrix
			with  entry given by $ \theta_{i}^{j-1} $. Based on this property, we obtain 
			\begin{align}
				\mathcal{I}_{k}=\frac{1}{ \prod_{i<j}^{q} \left(\gamma_{j}-\gamma_{i}\right)}\det\left(\bY_{k}\right),
			\end{align}
			where $ \bY_{k} $ is a $ q \times q $ matrix with entries
			\begin{align}\label{Yk}
				\!\!\!\!\{\bY_{k}\}_{m,n}=\!\begin{cases} \int_{0}^{\infty} x^{p-q+m+n-2}e^{-x}\mathrm{d}x,&\!\!\!\!\! ~~~n \ne k, \\
					\int_{0}^{\infty}	e^{-\hat{\lambda}/\gamma_{l}x}x^{p-M+m-2}e^{-x}\mathrm{d}x,&\!\! \!\!\!~~~n= k. \end{cases}	
			\end{align}
			Note that in \eqref{Yk}, we have made a change of variables, i.e., we have set $ x=c_{l}/\gamma_{l} $. The  integrals are evaluated based on \cite[Eq.3.381.4]{Gradshteyn2007} and \cite[Eq.3.471.9]{Gradshteyn2007} as
			\begin{align}
				\int_{0}^{\infty} x^{p-q+m+n-2}e^{-x}\mathrm{d}x&=\Gamma(p-q+m+n-1),\label{int1}\\
				\int_{0}^{\infty}	e^{-\hat{\lambda}/\gamma_{l}x}x^{\tilde{p}-1}e^{-x}\mathrm{d}x&=2 \left(\frac{\hat{\lambda}}{\gamma_{l}}\right)^{\!\!\frac{\tilde{p}}{2}}\mathrm{K}_{\tilde{p}}(2\sqrt{\hat{\lambda}/\gamma_{l}}),\label{int2}
			\end{align}
						where $ \tilde{p}=p-M+m-1 $, $ K_v(\cdot) $ is
			the modified Bessel function of the second kind \cite[Eq. 8.432.6]{Gradshteyn2007}. Substitution of \eqref{int1} and \eqref{int2} into \eqref{Yk}  together with the application of Laplace's expansion provides the  desired result.
		\end{appendices}
		\bibliographystyle{IEEEtran}
		
		\bibliography{IEEEabrv,mybib}
	\end{document}